\documentclass[a4paper]{elsarticle}
\usepackage[left=2.5cm,right=2.5cm]{geometry}

\usepackage[T1]{fontenc}
\usepackage{xcolor}

\usepackage{amssymb,amsmath,amsthm}
\usepackage{stmaryrd}
\usepackage{mleftright}
\usepackage{enumitem}

\usepackage{hyperref}
\usepackage{orcidlink}
\usepackage{url}
\hypersetup{colorlinks}

\usepackage{graphicx}

\usepackage[capitalize]{cleveref}
\usepackage{listings}

\usepackage{tkz-graph}
\usetikzlibrary{arrows}
\usetikzlibrary{arrows.meta}
\usepackage{diagbox}
\usepackage{accents}

\newtheorem{definition}{Definition}
\newtheorem{theorem}{Theorem}
\newtheorem{lemma}{Lemma}
\newtheorem{corollary}{Corollary}
\newtheorem{proposition}{Proposition}
\newtheorem{observation}{Observation}
\newtheorem{example}{Example}
\newtheorem{remark}{Remark}

\crefname{observation}{Observation}{Observations}
\Crefname{observation}{Observation}{Observations}
\crefname{proposition}{Proposition}{Propositions}
\Crefname{proposition}{Proposition}{Propositions}
\crefname{corollary}{Corollary}{Corollaries}
\Crefname{corollary}{Corollary}{Corollaries}
\crefname{figure}{Figure}{Figures}
\Crefname{figure}{Figure}{Figures}

\newcommand{\URSsymbol}{\mathsf{URS}}
\newcommand{\kURS}[1]{k_{\URSsymbol}\mleft(#1\mright)}
\newcommand{\URSmostgeneral}[3]{\URSsymbol\mleft(#1,#2,#3\mright)}
\newcommand{\URS}[2]{\URSsymbol{}(#1,#2)}
\newcommand{\URSreduced}[2]{\URSsymbol{}_{\mathrm{red}}\mleft(#1,#2\mright)}
\newcommand{\matrixrowsGroupInduced}[1]{\textsf{Mat}[#1]}
\newcommand{\range}[1]{[#1]}
\newcommand{\sfact}[1]{\mleft(#1\mright)!!}
\newcommand{\sfactexcludepars}[1]{#1!!}

\newcommand{\bif}[2]{F^{(#1,#2)}}
\newcommand{\bifAcomplete}[5]{F^{\mleft(#1,#2\mright)}_{[#3,#4]}\mleft(#5\mright)}
\newcommand{\bifAwithoutsymbs}[3]{F^{\mleft(#1,#2\mright)}\mleft(#3\mright)}
\newcommand{\bifMatfree}[4]{F^{\mleft(#1,#2\mright)}_{[#3,#4]}}

\newcommand{\colfreqvecSymbol}{\hspace{-0.4em}\rotatebox[origin=t]{-90}{\(\vec{}\)}f}
\newcommand{\colfreqvec}[1]{\colfreqvecSymbol{}^{\mleft(#1\mright)}}
\newcommand{\colfreq}[2]{\colfreqvecSymbol{}^{(#1)}_{#2}}

\newcommand{\rowfreqvecSymbol}{\vec{~}f}
\newcommand{\rowfreqvec}[1]{\rowfreqvecSymbol{}^{\mleft(#1\mright)}}
\newcommand{\rowfreq}[2]{\rowfreqvecSymbol{}^{\mleft(#1\mright)}_{#2}}

\newcommand{\card}[1]{\mleft|#1\mright|}
\newcommand{\transpose}[1]{{#1}\!^\top}
\newcommand{\identity}{\mathrm{id}}
\newcommand{\Aut}[1]{\operatorname{Aut}\mleft(#1\mright)}
\newcommand{\Dih}[1]{\operatorname{Dih}\mleft(#1\mright)}
\newcommand{\ord}[1]{\operatorname{ord}\mleft(#1\mright)}
\newcommand{\GL}[2]{\operatorname{GL}\mleft(#1,#2\mright)}
\renewcommand{\tilde}[1]{\accentset{\vspace{-0.06em}\sim}{#1}}

\newcommand{\symmetricgroup}[1]{S_{#1}}

\makeatletter
\def\ps@pprintTitle{
	\let\@oddhead\@empty
	\let\@evenhead\@empty
	\let\@oddfoot\@empty
	\let\@evenfoot\@oddfoot
}
\makeatother

\hypersetup{
	pdftitle = {Pairwise Reflection Symmetry in Frequency Rectangles},
	pdfauthor = {Enrico Iurlano and Günther R. Raidl},
	pdfkeywords = {Frequency rectangles, balanced combinatorial design, design of experiments, computational enumeration},
	pdfcreator = {LaTeX with hyperref},
	pdfproducer = {pdfLaTeX},
}

\AtBeginDocument{}

\begin{document}

\begin{frontmatter}
	\title{Pairwise Reflection Symmetry in Frequency Rectangles}
	\author{Enrico Iurlano\,\orcidlink{0000-0001-7528-0834}}
	\author{Günther R.~Raidl\,\orcidlink{0000-0002-3293-177X}\,}
	\affiliation{organization = {Algorithms and Complexity Group, TU Wien},
		country={Austria}
	}
	\begin{abstract}
		Many combinatorial designs ask for equal distribution of given symbols across the entries of a matrix. The paramount examples are Latin squares, where each symbol from $\{1,\dots,n\}$ appears once per row and column of an $n\times n$ matrix. Frequency rectangles extend this to $\lambda n \times \mu n$ matrices with repeated symbols under controlled column and row frequencies.
		In this more general setting, we examine structural properties of pairwise reflection-symmetry, which requires that, on every pair of columns, each ordered symbol pair $(p,q)$ occurs as often as its reversal $(q,p)$.
		This order-balance is precisely what makes head-to-head comparisons unbiased, i.e., no symbol gains a systematic advantage from the position it occupies relative to another, a fairness demand arising for instance when scheduling tournaments or laying out comparative trials.
		Existence of such objects for odd $\lambda$ turns out to be remarkably more subtle than for even $\lambda$.
		After showing that existence holds also for sufficiently large odd $\lambda$, we initiate the search for the smallest possible value of $\lambda$ in this setting.
		We obtain the insight that a per-column multiplicity of $\lambda=1$ can be achieved if and only if $n$ is a power of two.
		We complement the existence results with a direct product construction and add several further observations on the property.
		Finally, we propose and evaluate a quadratically constrained integer program to computationally search for these objects.
		The resulting experiments reveal that many of them possess an underlying group-theoretic structure which, as we conjecture, may even be unavoidable in certain settings.
	\end{abstract}
\end{frontmatter}

\section{Introduction\hphantom{\tiny\cite{andersen1980generalized,anthropic2026claude,bamberg2026association,bargachev2006some,boros1991existence,broder2000minwise,bruinedebruin2005save,cameron1976parallelisms,dummit2004abstract,gentle2023perfect,godsil2001algebraic,gurobi2025manual,heinrich1981maximum,hiner1970subsquare,hojny2025scip,itoh2000permutations,itoh2003onthesamplesize,iurlano2022scrambling,iurlano2023growth,jungnickel1987hyperfactorizations,keedwell2015latin,kiefer1958nonrandomized,macmahon1898new,markovski2005unbiased,mathon1999directed,na2023groupbased,oeis2026online,rosenberg1973number,tarui2003nearly,williams1949experimental,yuster2020perfect}}}\label{sec:introduction}
Beyond their combinatorial interest, Latin squares appear as building blocks in cryptographic primitives~\cite{keedwell2015latin}, statistical experimental design~\cite{williams1949experimental}, randomness extraction~\cite{markovski2005unbiased}, and a vast number of other fields~\cite{keedwell2015latin}.
We study the recently introduced property of \emph{pairwise reflection-symmetry}~\cite{iurlano2023growth} in frequency rectangles, i.e., matrices whose entries follow a uniform column and row distribution.
To convey the kind of symmetry (or fairness) to be studied, picture $n$ figure skaters competing in a series of trials, each trial assigning a starting order~$1,2,\ldots,n$ to the participants.
Suppose every skater is required to occupy every starting position exactly $\lambda$ times, which is already a uniformity condition on the column frequencies.
Crucially, as a trial unfolds the jury becomes progressively fatigued and increasingly subject to serial-position effects:
Early performances tend to be under-scored, later ones are anchored against the earlier ones, and the awarded marks therefore depend substantially on a skater's starting slot relative to the others~\cite{bruinedebruin2005save}.
In particular, a trial in which $p$ skates in position~$j_1$ and $q$ in position~$j_2$ is not interchangeable with one in which the two roles are swapped.
To make the cumulative head-to-head comparison of any two skaters $p$ and $q$ unbiased, one additionally insists that, for every pair of distinct positions~$j_1,j_2$, the number of trials in which $p$ skates from position~$j_1$ and $q$ from position~$j_2$ matches the number of trials in which the two roles are swapped.
Arranging the trials as the rows of a matrix~$A$ with entry~$A_{ij}$ recording the skater in position~$j$ of trial $i$, we can regard everything as an abstract matrix property.

A related idea is classical in statistical experimental design under the heading of \emph{residual} (or, in modern terminology, \emph{carry-over}) \emph{effects}:
In cross-over clinical and agricultural trials the residual effect of an earlier treatment may bias a later measurement, and one seeks designs in which every ordered pair of treatments appears in immediately consecutive positions equally often; these are the so-called \emph{Williams designs}~\cite{williams1949experimental}.
Pairwise reflection-symmetry is conceptually related, yet different in two respects:
It requires a certain balance not only between adjacent positions but between \emph{every} pair of columns, each of which must realize the symmetry in isolation; on the other hand, it alters the Latin square property that each column is a permutation by requiring $\lambda$-fold repetition of each symbol per column.
A parallel motivation for studying highly regular permutation collections comes from \emph{minwise independent families}~\cite{broder2000minwise} used in the prominent MinHash algorithm for rapid similarity estimation of documents in large datasets.
Since exact minwise independence requires impractically large families, one relaxes it to $k$-restricted (limited-independence) and other approximate variants:
The limited-independence relaxation and the minimal sample size realizing it are studied by Itoh et al.~\cite{itoh2000permutations,itoh2003onthesamplesize}, a nearly linear-size explicit $4$-minwise family from finite geometries is constructed by Tarui et al.~\cite{tarui2003nearly}, and group-theoretic properties of such families are exploited by Bargachev~\cite{bargachev2006some}.
A closely related object is the one of \emph{perfect sequence covering arrays}~\cite{yuster2020perfect}, for which group-based (decomposability) assumptions have proven particularly fruitful for constructions~\cite{gentle2023perfect,na2023groupbased}, just as for the more specific objects of directed $t$-packings and directed Steiner systems studied in earlier work~\cite{mathon1999directed}.

In fact pairwise reflection-symmetric matrices have been introduced because a strong relation to the latter objects is suspected and has been validated for some parameters of tractable size~\cite{iurlano2022scrambling}.
Given the vast number of concurring symmetry requirements we initiate a systematic study of the multiplicities $\lambda$ which enable existence of pairwise reflection-symmetric frequency rectangles.

\medskip

The paper is organized as follows.
\cref{sec:preliminaries} states the notation and recalls relevant concepts.
\cref{sec:constructions} contains our main existence and structural results, as well as the computer-aided existence study.
\cref{sec:conclusion} discusses emerging open questions.

\section{Preliminaries}\label{sec:preliminaries}
Let us start with the notation we use.
Throughout the paper, $n$, $\lambda$, $\mu$, $L$, and $M$ denote positive integers; we write $\mathbb{N}:=\{1,2,\ldots\}$ and abbreviate the integer interval $\range{n}:=\{1,2,\ldots,n\}$.
Matrices are typeset by capital letters $A,B,C,\ldots$ with the entry-indexing convention $A=(A_{ij})_{i\in\range{L},j\in\range{M}}$.
The symmetric group on $\range{n}$ is denoted $\symmetricgroup{n}$ and its identity permutation by $\identity\colon q\mapsto q$.
We will repeatedly identify a permutation~$\sigma\in\symmetricgroup{n}$ with the tuple $(\sigma(1),\sigma(2),\ldots,\sigma(n))\in\range{n}^n$, which lets us read each row of a suitably structured matrix~$A\in\range{n}^{L\times n}$ as a permutation.
By $\mathbb{Z}_n$ we denote the residue class ring modulo~$n$.
For a statement~$P$, the Iverson bracket $\llbracket P\rrbracket$ equals $1$ if $P$ holds and $0$ otherwise.

\bigskip

We will later encounter particular permutations in $\symmetricgroup{n}$:
An \emph{involution} $\pi\in\symmetricgroup{n}$ satisfies $\pi\circ\pi=\identity$ and a \emph{derangement} $\pi\in \symmetricgroup{n}$ is free of fixed points, i.e., $\pi(j)\neq j$ for all $j\in\range{n}$.
We recall that every permutation in $\symmetricgroup{n}$ can be decomposed into a finite number of cycles and that an involutory derangement consists exclusively of length-$2$ cycles.
A \emph{quasigroup} $(Q,\circ)$ is a set~$Q$ with a binary operation $\circ$ such that for all $a,b\in Q$ each of the equations $a\circ x=b$ and $y\circ a=b$ has a unique solution; if it is additionally associative, i.e., for all $a,b,c\in Q$, $(a\circ b)\circ c=a\circ(b\circ c)$, then it is a group.
If for all $a,b\in Q$ the commutation $a\circ b=b\circ a$ applies, we say that the quasigroup is \emph{commutative}; for groups we use the common term \emph{abelian}.
The \emph{Cayley table} of $(Q,\circ)$ is the array $T_Q\in Q^{Q\times Q}$ with $(T_Q)_{g,h}:=g\circ h$; after fixing a bijection from $Q$ onto $\range{n}$ with $n=\card{Q}$ one may view $T_Q$ as a matrix in $\range{n}^{n\times n}$.
A \emph{Latin square} of order~$n$ is a matrix~$A\in\range{n}^{n\times n}$ whose rows and columns correspond to permutations of $\range{n}$; we may also substitute $\range{n}$ by any other set of cardinality $n$.
The two notions are strongly related:
The Cayley table of a quasigroup is always a Latin square, and conversely, once the row, column, and symbol labels of a Latin square of order~$n$ are identified with a common $n$-element set~$Q$, setting $a\circ b$ to be the entry in row~$a$ and column~$b$ defines a unique quasigroup $(Q,\circ)$ whose Cayley table is that Latin square; the underlying quasigroup is a group exactly when its operation is associative.
Two such group-induced Latin squares will play an important role, firstly, the \emph{cyclic Latin square} of order~$n$, i.e., the Cayley table of $(\mathbb{Z}_n,+)$ with rows, columns, and symbols in their natural order, namely the matrix $A\in\range{n}^{n\times n}$ with entries $A_{ij}=((i+j-2)\bmod n)+1$; and, secondly, the Cayley table of the \emph{elementary abelian $2$-group} $\mathbb{Z}_2^k$ of order~$n=2^k$; see~\cite[p.~136]{dummit2004abstract} for the latter.

We are now ready to recall the generalization of Latin squares relevant to the paper.
\begin{definition}\label{def:frequency-rectangle}
	A matrix~$A\in\range{n}^{\lambda n \times \mu n}$ is called a \emph{frequency rectangle} if $(\card{\{i:A_{ij}=q\}})_{j=1}^{\mu n} = (\lambda)_{j=1}^{\mu n}$ and $(\card{\{j:A_{ij}=q\}})_{i=1}^{\lambda n}=(\mu)_{i=1}^{\lambda n}$ for all $q\in\range{n}$, i.e., the columnwise as well as the rowwise frequencies per symbol are equally distributed.
\end{definition}
Frequency rectangles extend the frequency squares ($\lambda=\mu$) going back to MacMahon~\cite[p.~276]{macmahon1898new}.
The concept can be traced back to special cases of the generalized Youden squares of Kiefer~\cite[p.~689]{kiefer1958nonrandomized}, namely those in which the number of symbols divides both side lengths, and, alternatively, of the generalized Latin rectangles of Andersen and Hilton~\cite{andersen1980generalized}, namely those whose entries are singletons rather than sets.
For the special case~$\lambda=\mu=1$ we recover the notion of a classical Latin square.
Latin squares can appear in \emph{reduced} form~\cite{keedwell2015latin}, i.e., the first row and first column agree with the tuple $(1,2,\ldots,n)$.
We can here straightforwardly extend this notion.
In \cref{fig:example-reflection-symmetry} such a reduced representative with $(\lambda,\mu)=(2,1)$ is depicted.
This reduced form enables the equivalent interpretation of a matrix as a multiset of permutations.
\begin{definition}\label{def:reduced-frequency-rectangle}
	We call, for $\mu=1$, a frequency rectangle \emph{reduced} if the first row is the tuple $(1,2,\dots,n)$ and the rows, seen as $n$-tuples, are in non-decreasing ordering with respect to the lexicographical order on $n$-tuples induced by $(\range{n},\leq)$.
\end{definition}
\begin{figure}[tb]
	\centering
	\scalebox{0.5}{
		\begin{tikzpicture}
			\definecolor{color1}{HTML}{5565ca}
			\definecolor{color2}{HTML}{c3c9ed}
			\definecolor{color3}{HTML}{f9ecf4}
			\definecolor{color4}{HTML}{f695d2}
			\definecolor{color5}{HTML}{e567af}
			\def\dt{{
						1,2,3,4,5,.,.,.,.,.,
						1,3,2,5,4,.,.,.,.,.,
						2,1,5,4,3,.,.,.,.,.,
						2,5,1,3,4,.,.,.,.,.,
						3,1,4,5,2,.,.,.,.,.,
						3,4,1,2,5,.,.,.,.,.,
						4,3,5,2,1,.,.,.,.,.,
						4,5,3,1,2,.,.,.,.,.,
						5,2,4,3,1,.,.,.,.,.,
						5,4,2,1,3,.,.,.,.,.
					}}
			\newcounter{colorval}
			\foreach \i in {0,1,...,9} {
					\foreach \j in {0,1,2,3,4} {
							\pgfmathsetcounter{colorval}{\dt[10*(10-1-\i)+\j]}
							\if \j 1
								\fill[color\thecolorval] (\j-0.5,\i-0.5) rectangle ++(1,1);
							\fi
							\if \j 4
								\fill[color\thecolorval] (\j-0.5,\i-0.5) rectangle ++(1,1);
							\fi
							\node[] at (\j,\i) {\Huge\thecolorval};
							\draw[black] (\j-0.5,\i-0.5) rectangle ++(1,1);
						}
				}
			\def\dttwo{{
						1,2,3,4,5,.,.,.,.,.,
						1,3,2,5,4,.,.,.,.,.,
						2,1,5,4,3,.,.,.,.,.,
						2,5,1,3,4,.,.,.,.,.,
						3,1,4,5,2,.,.,.,.,.,
						5,2,4,3,1,.,.,.,.,.,
						3,4,1,2,5,.,.,.,.,.,
						4,3,5,2,1,.,.,.,.,.,
						5,4,2,1,3,.,.,.,.,.,
						4,5,3,1,2,.,.,.,.,.
					}}
			\foreach \i in {0,1,...,9} {
					\def\xshift{0}
					\if \i 4
						\def\xshift{-0.075}
					\fi
					\if \i 3
						\def\xshift{0.075}
					\fi
					\if \i 2
						\def\xshift{0.02}
					\fi
					\if \i 1
						\def\xshift{0.0}
					\fi
					\if \i 0
						\def\xshift{-0.085}
					\fi
					\foreach \j in {0,1,2,3,4} {
							\pgfmathsetcounter{colorval}{\dttwo[10*(10-1-\i)+\j]}
							\if \j 1
								\fill[color\thecolorval] (\j-0.5+7.65+\xshift,\i-0.5) rectangle ++(1,1);
							\fi
							\if \j 4
								\fill[color\thecolorval] (\j-0.5+7.65+\xshift,\i-0.5) rectangle ++(1,1);
							\fi
							\draw[black] (\j-0.5+7.65+\xshift,\i-0.5) rectangle ++(1,1);
							\node[] at (\j+7.65+\xshift,\i) {\Huge \thecolorval};
						}
				}
			\foreach \i in {0,1,...,9} {
					\def\xshift{0}
					\if \i 4
						\def\xshift{-0.075}
					\fi
					\if \i 3
						\def\xshift{0.075}
					\fi
					\if \i 2
						\def\xshift{0.02}
					\fi
					\if \i 1
						\def\xshift{0.0}
					\fi
					\if \i 0
						\def\xshift{-0.085}
					\fi

					\foreach \j in {0,1,2,3,4} {
							\pgfmathsetcounter{colorval}{\dttwo[10*(10-1-\i)+\j]}

							\if \j 1
								\fill[color\thecolorval] (\j-0.5+7.65+\xshift,\i-0.5) rectangle ++(1,1);
							\fi
							\if \j 4
								\fill[color\thecolorval] (\j-0.5+7.65+\xshift,\i-0.5) rectangle ++(1,1);
							\fi

							\draw[black] (\j-0.5+7.65+\xshift,\i-0.5) rectangle ++(1,1);
							\node[] at (\j+7.65+\xshift,\i) {\Huge \thecolorval};
						}
					\if \i 0 \def\sourcey{2} \fi
					\if \i 1 \def\sourcey{0} \fi
					\if \i 2 \def\sourcey{3} \fi
					\if \i 3 \def\sourcey{4} \fi
					\if \i 4 \def\sourcey{1} \fi
					\ifnum\i<5
						\pgfmathsetmacro{\xs}{4.65}
						\pgfmathsetmacro{\ys}{\sourcey}
						\pgfmathsetmacro{\xe}{7+\xshift}
						\pgfmathsetmacro{\ye}{\i}
						\pgfmathsetmacro{\cxa}{\xs + 0.8}
						\pgfmathsetmacro{\cya}{\ys}
						\pgfmathsetmacro{\cxb}{\xe - 0.8}
						\pgfmathsetmacro{\cyb}{\ye}
						\draw[->, thick]
						(\xs,\ys)
						.. controls (\cxa,\cya) and (\cxb,\cyb)
						.. (\xe,\ye);
					\fi
				}
			\def\dtthree{{
						2,5,.,.,.,.,.,.,.,.,
						3,4,.,.,.,.,.,.,.,.,
						1,3,.,.,.,.,.,.,.,.,
						5,4,.,.,.,.,.,.,.,.,
						1,2,.,.,.,.,.,.,.,.,
						2,1,.,.,.,.,.,.,.,.,
						4,5,.,.,.,.,.,.,.,.,
						3,1,.,.,.,.,.,.,.,.,
						4,3,.,.,.,.,.,.,.,.,
						5,2,.,.,.,.,.,.,.,.
					}}
			\foreach \i in {0,1,...,9} {
					\foreach \j in {0,1} {
							\pgfmathsetcounter{colorval}{\dtthree[10*(10-1-\i)+\j]}
							\if \j 0
								\fill[color\thecolorval] (\j-0.5+15,\i-0.5) rectangle ++(1,1);
							\fi
							\if \j 1
								\fill[color\thecolorval] (\j-0.5+15,\i-0.5) rectangle ++(1,1);
							\fi
							\ifnum \i <5
								\node[rotate=180] at (\j+15,\i) {\Huge \thecolorval};
							\else
								\node[rotate=0] at (\j+15,\i) {\Huge \thecolorval};
							\fi
							\draw[black] (\j-0.5+15,\i-0.5) rectangle ++(1,1);
						}
				}
			\begin{scope}[xshift=4cm]
				\def\dtthree{{
							2,5,.,.,.,.,.,.,.,.,
							3,4,.,.,.,.,.,.,.,.,
							1,3,.,.,.,.,.,.,.,.,
							5,4,.,.,.,.,.,.,.,.,
							1,2,.,.,.,.,.,.,.,.,
							2,1,.,.,.,.,.,.,.,.,
							4,5,.,.,.,.,.,.,.,.,
							3,1,.,.,.,.,.,.,.,.,
							4,3,.,.,.,.,.,.,.,.,
							5,2,.,.,.,.,.,.,.,.
						}}
				\foreach \i in {0,1,...,9} {
						\foreach \j in {0,1} {
								\pgfmathsetcounter{colorval}{\dtthree[10*(10-1-\i)+\j]}
								\if \j 0
									\fill[color\thecolorval] (\j-0.5+15,\i-0.5) rectangle ++(1,1);
								\fi
								\if \j 1
									\fill[color\thecolorval] (\j-0.5+15,\i-0.5) rectangle ++(1,1);
								\fi
								\ifnum \i <5
									\node[rotate=180] at (\j+15,\i) {\Huge};
								\else
									\node[rotate=0] at (\j+15,\i) {\Huge};
								\fi
								\draw[black] (\j-0.5+15,\i-0.5) rectangle ++(1,1);
							}
					}
				\node[circle,fill=black,minimum size=1pt] at (15.5,4.5) [label=below right:{\Large ~}]{};
			\end{scope}
		\end{tikzpicture}
	}
	\caption{Pairwise reflection-symmetry for an example $A\in\range{5}^{10\times5}$ at the $10\times2$ submatrix spanned by $(j_1,j_2)=(2,5)$.
		The involution mapping each $1\times2$ row to an own flipped $1\times2$ row is visualized by the reflection through the barycenter (after the row-reordering indicated by arrows).
	}\label{fig:example-reflection-symmetry}
\end{figure}
First, let us introduce and study the concept of pairwise reflection-symmetry~\cite{iurlano2023growth} under as weak as possible assumptions; originally it was motivated even earlier in the unpublished work~\cite{iurlano2022scrambling}.
Afterwards we will mostly deal with this property under an additional assumption of uniformity, in particular implying that all symbols of the initially given alphabet $\range{n}$ must appear among the entries.
\begin{definition}\label{def:pairwise-reflection-symmetry}
	We say that $A\in \range{n}^{L\times M}$ is \emph{pairwise reflection-symmetric} if for all $(j_1,j_2)\in\range{M}^2$ with $j_1\neq j_2$ we have for any $(p,q)\in\range{n}^2$ with $p\neq q$ that
	\begin{equation*}
		\card{\{i:A_{i,j_1}=p, A_{i,j_2}=q\}}=\card{\{i:A_{i,j_1}=q, A_{i,j_2}=p\}}.
	\end{equation*}
\end{definition}
The latter definition enforces that the \emph{bivariate frequency matrix} associated to each column-pair $(j_1,j_2)$ is symmetric, i.e.,
\begin{equation*}
	\bifAwithoutsymbs{j_1}{j_2}{A}=\transpose{\left(\bifAwithoutsymbs{j_1}{j_2}{A}\right)},
\end{equation*}
with the matrix~$\bifAwithoutsymbs{j_1}{j_2}{A}=(\bifAcomplete{j_1}{j_2}{p}{q}{A})_{p,q\in\range{n}}$ made up of entries
\begin{equation*}
	\bifAcomplete{j_1}{j_2}{p}{q}{A}:=\card{\{i:A_{i,j_1}=p, A_{i,j_2}=q\}}.
\end{equation*}
A concrete example of a pairwise reflection-symmetric matrix in the sense of \cref{def:pairwise-reflection-symmetry} is displayed in \cref{fig:example-reflection-symmetry}; an example of an associated bivariate frequency matrix is given in \cref{exa:bivariate-frequency-matrix} below.
When referring to these matrices, we suppress the matrix~$A$, column-indices, or symbol-pair in question when they are clear from the context.
Similarly, for $A$ let us define its $j$-th \emph{column-frequency vector} $\colfreqvec{j}=(\colfreq{j}{q})_{q\in\range{n}}$ by the entries~$\colfreq{j}{q} := \card{\{i:A_{ij}=q\}}$.
Similarly, call $\rowfreqvec{i} = (\rowfreq{i}{q})_{q\in\range{n}}$ with $\rowfreq{i}{q} :=\card{\{j:A_{ij}=q\}}$ its $i$-th \emph{row-frequency vector}.
The column-frequency vectors store therefore the marginals of the bivariate frequency matrices, i.e., the $p$-th row sum of $\bifAwithoutsymbs{j_1}{j_2}{A}$ equals $\colfreq{j_1}{p}$, whereas its $q$-th column sum equals $\colfreq{j_2}{q}$.
\begin{remark}
	The name ideates from the fact that for each column-pair $(j_1,j_2)$ the required reflection to be met has the geometrically meaningful association visualized in \cref{fig:example-reflection-symmetry}.
\end{remark}
\begin{example}\label{exa:bivariate-frequency-matrix}
	The bivariate frequency matrix for the example in \cref{fig:example-reflection-symmetry} at column-pair $(j_1,j_2)=(2,5)$ is given by the following:
	\begin{equation*}
		\begin{array}{c|ccccc}
			$\diagbox{p}{q}$ & 1 & 2 & 3 & 4 & 5 \\\hline
			1                & 0 & 1 & 1 & 0 & 0 \\
			2                & 1 & 0 & 0 & 0 & 1 \\
			3                & 1 & 0 & 0 & 1 & 0 \\
			4                & 0 & 0 & 1 & 0 & 1 \\
			5                & 0 & 1 & 0 & 1 & 0
		\end{array}
	\end{equation*}
\end{example}
\begin{observation}\label{obs:inheritance-column-distribution}
	If $A$ is pairwise reflection-symmetric, then for any choice of columns~$j_1\neq j_2$ we have $\colfreqvec{j_1}=\colfreqvec{j_2}$.
	In particular, the entry of the column-frequency associated to symbol~$q$ is independent of the column index.\hfill\qedsymbol
\end{observation}
\begin{definition}
	Let us denote by $\URSmostgeneral{n}{\lambda}{\mu}$ the class of all pairwise reflection-symmetric frequency rectangles $A\in\range{n}^{\lambda n\times\mu n}$.
\end{definition}
In the case $\mu=1$ we simply write $\URS{n}{\lambda}:=\URSmostgeneral{n}{\lambda}{1}$ and, for brevity, speak about ``an $\URS{n}{\lambda}$'' when referring to ``a matrix in $\URS{n}{\lambda}$''.
We will exclusively focus on $\URS{n}{\lambda}$, i.e., the scenario $\mu=1$ where each row therefore corresponds to a permutation of $\range{n}$.
\begin{observation}\label{obs:shuffling-invariance-pairwise-reflection-symmetry}
	Permuting rows of an $\URS{n}{\lambda}$ preserves pairwise reflection-symmetry, and permuting its columns by some $\pi\in \symmetricgroup{n}$ replaces $\bif{j_1}{j_2}$ by $\bif{\pi^{-1}(j_1)}{\pi^{-1}(j_2)}$, therefore also preserves pairwise reflection-symmetry.
	We can bring an $\URS{n}{\lambda}$ into reduced form by first fixing a (necessarily existing) row starting with $1$ and permuting only the last $n-1$ columns of the matrix so that this row then reads $(1,2,\ldots,n)$; afterward we sort the rows lexicographically.\hfill$\square$
\end{observation}
\begin{observation}\label{obs:row-count-necessarily-even}
	Any $A\in\URS{n}{\lambda}$ with $n\geq2$ has an even row-count $\lambda n$.
\end{observation}
\begin{proof}
	Each row of the submatrix spanned by the column-pair $(j_1,j_2)$ consists of two distinct symbols, such that each of its rows can be brought into correspondence (technically an involution) with its own (reversed) row; see \cref{fig:example-reflection-symmetry}.
\end{proof}
\begin{lemma}\label{lem:family-of-permutations-plus-reflection-symmetric-implies-urs}
	Assume $A\in\range{n}^{\lambda n\times n}$ has just permutations of $\range{n}$ as rows, i.e., $\rowfreqvec{i}=(1)_{q\in\range{n}}$ for $i=1,\dots,\lambda n$.
	Then $A$ is pairwise reflection-symmetric iff $A\in\URS{n}{\lambda}$.
\end{lemma}
\begin{proof}
	Due to $A\in\range{n}^{\lambda n\times n}$, we have $\lambda n$ occurrences of $q$ among all the entries of $A$, for each $q\in\range{n}$.
	Assume first $A$ is pairwise reflection-symmetric.
	If now symbol~$q$ was present in column~$1$ more (respectively less) than $\lambda$ times, then by pairwise reflection symmetry each other column, by \cref{obs:inheritance-column-distribution}, would also contain more (respectively less) than $\lambda$ times the entry~$q$.
	By counting the occurrences, we would obtain a contradiction on the number of $q$-valued entries throughout the entire matrix~$A$.
	The opposite proof direction is trivial.
\end{proof}
The following \cref{pro:set-interpretability-lambda-one-two-odd-n} is handy as it allows us to regard certain pairwise reflection-symmetric matrices as \emph{subsets} of $\symmetricgroup{n}$ of cardinality~$2n$ without losing generality.
\begin{proposition}\label{pro:set-interpretability-lambda-one-two-odd-n}
	The following three assertions are valid.
	\begin{enumerate}[label=(\roman*)]
		\item Every $\URS{n}{1}$ consists of $n$ pairwise distinct rows.\label{ite:set-interpretability-lambda-one}
		\item For each odd $n\geq3$, every $\URS{n}{2}$ consists of $2n$ pairwise distinct rows.\label{ite:set-interpretability-lambda-two}
		\item Let $n\geq3$ be odd, $A\in\URS{n}{\lambda}$, and suppose there are precisely $\rho$ rows hosting the same permutation.
		      Then we have $\rho\leq\lfloor\lambda n/(n+3)\rfloor<\lambda$.\label{ite:set-interpretability-lambda-general}
	\end{enumerate}
\end{proposition}
\begin{proof}
	Property~\ref{ite:set-interpretability-lambda-one} follows immediately from the fact that $\URS{n}{1}$ consists of particular Latin squares and validity of~\ref{ite:set-interpretability-lambda-two} is a direct consequence of~\ref{ite:set-interpretability-lambda-general} which, for $\lambda=2$, affirms $\rho\leq1=\lfloor2n/(n+3)\rfloor$ for $n\geq 3$.

	Let us address~\ref{ite:set-interpretability-lambda-general}.
	For a permutation~$\sigma$ write $f(\sigma)$ for the number of fixed points of $\sigma$ and $t(\sigma)$ for the number of length-$2$ cycles present in $\sigma$.
	Without loss of generality assume that precisely the first $\rho$ rows correspond to the identity permutation.
	Denote by $T$ the total number of length-$2$ cycles among all rows (the identity rows contribute none).
	By uniform column-frequency $\lambda n$ is the number of fixed points present in an $\URS{n}{\lambda}$; a number of $\rho n$ of them appears already in the first $\rho$ rows, meaning that $(\lambda-\rho)n$ of them must be contained in the rest of non-identity permutations.
	Consequently, in this rest, there can be at most $(\lambda-\rho)n$ permutations having at least one fixed point.
	Complementarily, the number of derangements must be at least $(\lambda n-\rho)-(\lambda-\rho)n=\rho(n-1)$.

	The $n-f(\sigma)-2t(\sigma)$ points lying outside the fixed points and length-$2$ cycles of $\sigma$ sit in cycles of length at least $3$.
	Summing over the non-identity rows~$\sigma$, we have
	\begin{equation*}
		\sum_{\sigma\neq\identity}\bigl(n-f(\sigma)-2t(\sigma)\bigr)=(\lambda n-\rho)n-(\lambda-\rho)n-2T=n\lambda(n-1)-2T.
	\end{equation*}
	Since $n$ is odd, every row hosting a derangement of an odd-sized set and hence containing a cycle of length at least $3$, contributes a summand of at least $3$ to the left-hand side.
	As we earlier observed there are at least $\rho(n-1)$ such rows and we obtain $n\lambda(n-1)-2T\geq 3\rho(n-1)$, equivalently meaning that $T\leq(n-1)(\lambda n-3\rho)/2$.

	Now we also use the fact that the structure of the first $\rho$ rows forces the existence of at least $T\geq\rho\binom{n}{2}=\rho n(n-1)/2$ counterbalancing length-$2$ cycles throughout all rows:
	We chain the two estimates on $T$ to obtain after simplification that $\lambda n-3\rho\geq \rho n$ meaning $\rho \leq \lambda n/(n+3)$.
\end{proof}
In the following we give another insight linking $\URS{n}{\lambda}$ and perfect sequence covering arrays.
\begin{definition}[adapted from~\cite{yuster2020perfect}]\label{def:psca}
	A matrix~$A\in\range{n}^{L\times n}$ is called a \emph{Perfect Sequence Covering Array} (PSCA) of \emph{strength} $k$ if for all $(\sigma_1,\ldots,\sigma_k)\in \range{n}^k$ with pairwise different entries we have that $\{i\in\range{L}:(\sigma_1,\ldots,\sigma_k)\text{ }\allowbreak{}\text{is}\text{ a}\text{ subsequence }\allowbreak{}\text{of }(A_{ij})_{j\in \range{n}}\}$ has cardinality~$L/k!$.
	Here, $(\sigma_1,\ldots,\sigma_k)$ is a \emph{subsequence} of a row $(a_1,\ldots,a_n)$ if there are indices $j_1<j_2<\cdots<j_k$ with $a_{j_m}=\sigma_m$ for all $m\in\range{k}$, i.e., the symbols occur in the prescribed relative order but not necessarily in consecutive positions.
\end{definition}
\begin{observation}\label{obs:urs-is-strength-two}
	Every $A\in\URS{n}{\lambda}$ is a PSCA of strength $2$ whose \emph{row-inverse} (the matrix obtained by replacing each row with its inverse permutation) is again a PSCA of strength $2$.
\end{observation}
\begin{proof}
	Fix symbols $p\neq q$. Summing the defining equalities $\bifMatfree{j}{k}{p}{q}=\bifMatfree{j}{k}{q}{p}$ over all positions $1\leq j<k\leq n$ shows that the number of rows in which $p$ precedes $q$ equals the number of rows in which $q$ precedes $p$.
	Since every row is a permutation of $\range{n}$, these two counts sum to $L=\lambda n$, so each equals $L/2$.
	For the row-inverse, fix positions $j<k$ and sum the same equalities over all symbols $1\leq p<q\leq n$ instead: The number of rows with $A_{ij}<A_{ik}$ equals the number of rows with $A_{ij}>A_{ik}$, and again both sum to~$L$.
	Since row~$i$ of the row-inverse lists $j$ before $k$ exactly when $A_{ij}<A_{ik}$, this is the strength-$2$ requirement for the row-inverse.
\end{proof}

\section{Constructions}\label{sec:constructions}
Let us first constructively argue the existence of representatives for $\lambda=2$.
\begin{lemma}\label{lem:vertically-stacking-latin-squares-yields-reflection-symmetry}
	For the cyclic Latin square $A\in\range{n}^{n\times n}$ and its version with reversed rows, $A'=(A_{i,n-j+1})_{i,j=1}^n$, their vertical stacking $\binom{A}{A'}\in\range{n}^{2n\times n}$ is pairwise reflection-symmetric.
\end{lemma}
\begin{proof}
	Recall that $A_{ij}=((i+j-2)\bmod n)+1$, see \cref{sec:preliminaries}.
	Then the column-difference $A_{i,j_2}-A_{i,j_1}\equiv j_2-j_1\pmod n$ is independent of $i$.
	As $i$ ranges over $\range{n}$, the first coordinate $A_{i,j_1}$ takes each value exactly once.
	Therefore, $\bifAcomplete{j_1}{j_2}{p}{q}{A}=\llbracket q-p\equiv j_2-j_1\ (\mathrm{mod}\ n)\rrbracket$.
	By row-reversal, $\bifAcomplete{j_1}{j_2}{p}{q}{A'}=\bifAcomplete{n-j_1+1}{n-j_2+1}{p}{q}{A}=\llbracket q-p\equiv -(j_2-j_1)\ (\mathrm{mod}\ n)\rrbracket$.
	Summing, we obtain the bivariate frequency matrix
	\begin{equation*}
		\bifAcomplete{j_1}{j_2}{p}{q}{\binom{A}{A'}}=\llbracket q-p\equiv j_2-j_1\ (\mathrm{mod}\ n)\rrbracket+\llbracket q-p\equiv -(j_2-j_1)\ (\mathrm{mod}\ n)\rrbracket,
	\end{equation*}
	which is symmetric per column-pair $(j_1,j_2)$, as exchanging $p$ and $q$ interchanges the two brackets.
\end{proof}
\begin{corollary}\label{cor:even-lambda-setting-is-simple}
	For arbitrary $n$ and \emph{even} $\lambda\geq 2$ there exists an~$\URS{n}{\lambda}$.
\end{corollary}
\begin{proof}
	Vertically stack $\lambda/2$ copies of the $2n\times n$ matrix~$\binom{A}{A'}$ in \cref{lem:vertically-stacking-latin-squares-yields-reflection-symmetry}.
\end{proof}
By \cref{obs:row-count-necessarily-even}, an even row-count $\lambda n$ is necessary for existence.
Hence a natural question emerging from \cref{cor:even-lambda-setting-is-simple} is whether for \emph{odd} $\lambda$ and even $n$ this is the case, too.
Its answer will turn out to be subtle and somewhat unexpected.
Let us here anticipate selected results of the computational enumeration in \cref{sec:counting-urs}.
While for $n\in\{2,4,8\}$ representatives in $\URS{n}{1}$ exist, this is not the case for $n\in\{6,10,12\}$.
However, we find representatives for $n=6$ with respect to the successive odd candidate multiplicity~$\lambda=3$; for $n=10$ deciding (non)existence remained beyond our computational reach.
Let us provide some notation to inspect the phenomenon more systematically and denote, for an even $n$, by $\kURS{n}$ the smallest \emph{odd multiplicity} $\lambda$ such that an $\URS{n}{\lambda}$ exists.
Well-definedness of $\kURS{n}$ follows from the postponed \cref{cor:involutory-derangement-construction}.
The following gives a ``more delicate'' counterpart to \cref{cor:even-lambda-setting-is-simple} for odd multiplicities.
\begin{corollary}\label{cor:odd-lambda-above-kurs}
	For \emph{even} $n$ and any \emph{odd} $\lambda\geq \kURS{n}$ there exists an~$\URS{n}{\lambda}$.
\end{corollary}
\begin{proof}
	For $\lambda=\kURS{n}$ nothing is to show.
	For $\lambda=\kURS{n}+2r$ take $r$ copies of a $\URS{n}{2}$ (existing by \cref{cor:even-lambda-setting-is-simple}) and stack it with a $\URS{n}{\kURS{n}}$.
\end{proof}
In order to prove the next result we recall two concepts.
For a positive integer $n$ we understand under the \emph{semifactorial} $\sfactexcludepars{n}$ of $n$ the product of all positive integers that have the parity of $n$ and do not exceed $n$, i.e.,
\begin{equation*}
	\sfactexcludepars{n}:=n(n-2)(n-4)\cdots(2-(n\bmod2)).
\end{equation*}
Throughout, we additionally use the convention $\sfact{-1}:=1$.
The semifactorial characterizes the number of \emph{involutory derangements} (involutions that are derangements, see \cref{sec:preliminaries}).
When $n$ is even, the number of such permutations is
$\sfact{n-1}$, otherwise, when $n$ is odd, no such permutations exist.
For even $n$, a \emph{hyperfactorization} of index~$c$ of the complete graph~$K_n$ is a multiset~$\mathcal{M}$ of perfect matchings of $K_n$ such that every pair of disjoint edges is contained in exactly $c$ members of~$\mathcal{M}$~\cite{jungnickel1987hyperfactorizations}.
In the example in \cref{fig:k6-perfect-matchings}, the index is $c=1$ where, e.g., the disjoint edges $\{1,2\}$ and $\{5,6\}$ appear jointly only in the matching at the top left.
We identify a perfect matching with the involutory derangement that swaps the endpoints of each of its edges; in this way, the involutory derangements of $\symmetricgroup{n}$ correspond bijectively to the perfect matchings of~$K_n$.
An anonymous referee pointed out to us that the construction of \cref{cor:involutory-derangement-construction} below is a special case of the following more general approach based on hyperfactorizations.
\begin{theorem}\label{thm:hyperfactorization-bound}
	Let $n\geq4$ be even and let $\mathcal{M}$ be a hyperfactorization of index~$c$ of $K_n$.
	Stacking $r:=c(n-3)$ copies of $\identity$ and the involutory derangements corresponding to the members of~$\mathcal{M}$ yields a member of $\URS{n}{r}$.
\end{theorem}
\begin{proof}
	Every edge lies in exactly $r=c(n-3)$ members of~$\mathcal{M}$: Fixing an edge~$e$, each of the $\binom{n-2}{2}$ edges disjoint from~$e$ lies together with~$e$ in $c$ members, each of which contains $e$ and $n/2-1$ such edges.
	Double counting edges gives $\card{\mathcal{M}}=r(n-1)$, so the stacked matrix has $rn$ rows, all of them permutations of~$\range{n}$.
	By \cref{lem:family-of-permutations-plus-reflection-symmetric-implies-urs} it suffices to verify pairwise reflection-symmetry.
	Fix positions $j\neq k$ and symbols $p\neq q$.
	If $(p,q)=(j,k)$, exactly the $r$ identity rows realize $(p,q)$ in positions $(j,k)$; if $(p,q)=(k,j)$, exactly the $r$ involutory derangements swapping $j$ and $k$, i.e., the members containing the edge~$\{j,k\}$, do so.
	If $\card{\{j,k,p,q\}}=4$, the rows realizing $(p,q)$ in positions $(j,k)$ are the involutory derangements whose matchings contain the disjoint edges $\{j,p\},\{k,q\}$, and those realizing $(q,p)$ contain $\{j,q\},\{k,p\}$; both numbers equal~$c$.
	If $\card{\{j,k,p,q\}}=3$, no row realizes $(p,q)$ or $(q,p)$ in positions $(j,k)$, since an involutory derangement~$\sigma$ has no fixed points and $\sigma(j)=k$ forces $\sigma(k)=j$, and an identity row only realizes $(j,k)$.
\end{proof}
\begin{corollary}\label{cor:involutory-derangement-construction}
	For every even $n\geq2$ there exists an $\URS{n}{\sfact{n-3}}$, consisting of $\sfact{n-3}$ copies of~$\identity$ and all $\sfact{n-1}$ involutory derangements of~$\symmetricgroup{n}$; the multiplicity $\sfact{n-3}$ is odd.
\end{corollary}
\begin{proof}
	For $n=2$, the matrix consists of the rows $(1,2)$ and $(2,1)$.
	For $n\geq4$, let $\mathcal{M}$ be the multiset consisting of all $\sfact{n-1}$ perfect matchings of~$K_n$, each with multiplicity one; under the above identification, its members correspond exactly to the involutory derangements of~$\symmetricgroup{n}$.
	Two disjoint edges of~$K_n$ extend to a perfect matching in exactly $\sfact{n-5}$ ways, namely by a perfect matching of the remaining $n-4$ vertices, so they are contained in exactly $\sfact{n-5}$ members of~$\mathcal{M}$.
	Hence $\mathcal{M}$ is a hyperfactorization of index~$c=\sfact{n-5}$, and \cref{thm:hyperfactorization-bound} applies with $r=\sfact{n-5}(n-3)=\sfact{n-3}$.
\end{proof}
\begin{figure}[tb]
	\centering
	\definecolor{identityblockshade}{HTML}{9AD5C0}
	\begin{tikzpicture}[thick,bv/.style={circle,fill=black,inner sep=0pt,minimum size=4pt}]
		\begin{scope}[shift={(0.000,0.00)}]
			\fill[black!7,rounded corners=4pt] (-0.86,-1.02) rectangle (0.86,1.02);
			\foreach \v in {1,...,6} \node[bv,label={[label distance=-0.5mm]{90-60*(\v-1)}:\scriptsize$\v$}] (v\v) at ({90-60*(\v-1)}:0.55) {};
			\draw (v1) -- (v2) (v3) -- (v4) (v5) -- (v6);
		\end{scope}
		\begin{scope}[shift={(1.950,0.00)}]
			\fill[black!7,rounded corners=4pt] (-0.86,-1.02) rectangle (0.86,1.02);
			\foreach \v in {1,...,6} \node[bv,label={[label distance=-0.5mm]{90-60*(\v-1)}:\scriptsize$\v$}] (v\v) at ({90-60*(\v-1)}:0.55) {};
			\draw (v1) -- (v2) (v3) -- (v5) (v4) -- (v6);
		\end{scope}
		\begin{scope}[shift={(3.900,0.00)}]
			\fill[black!7,rounded corners=4pt] (-0.86,-1.02) rectangle (0.86,1.02);
			\foreach \v in {1,...,6} \node[bv,label={[label distance=-0.5mm]{90-60*(\v-1)}:\scriptsize$\v$}] (v\v) at ({90-60*(\v-1)}:0.55) {};
			\draw (v1) -- (v2) (v3) -- (v6) (v4) -- (v5);
		\end{scope}
		\begin{scope}[shift={(5.850,0.00)}]
			\fill[black!7,rounded corners=4pt] (-0.86,-1.02) rectangle (0.86,1.02);
			\foreach \v in {1,...,6} \node[bv,label={[label distance=-0.5mm]{90-60*(\v-1)}:\scriptsize$\v$}] (v\v) at ({90-60*(\v-1)}:0.55) {};
			\draw (v1) -- (v3) (v2) -- (v4) (v5) -- (v6);
		\end{scope}
		\begin{scope}[shift={(7.800,0.00)}]
			\fill[black!7,rounded corners=4pt] (-0.86,-1.02) rectangle (0.86,1.02);
			\foreach \v in {1,...,6} \node[bv,label={[label distance=-0.5mm]{90-60*(\v-1)}:\scriptsize$\v$}] (v\v) at ({90-60*(\v-1)}:0.55) {};
			\draw (v1) -- (v3) (v2) -- (v5) (v4) -- (v6);
		\end{scope}
		\begin{scope}[shift={(0.000,-2.20)}]
			\fill[black!7,rounded corners=4pt] (-0.86,-1.02) rectangle (0.86,1.02);
			\foreach \v in {1,...,6} \node[bv,label={[label distance=-0.5mm]{90-60*(\v-1)}:\scriptsize$\v$}] (v\v) at ({90-60*(\v-1)}:0.55) {};
			\draw (v1) -- (v3) (v2) -- (v6) (v4) -- (v5);
		\end{scope}
		\begin{scope}[shift={(1.950,-2.20)}]
			\fill[black!7,rounded corners=4pt] (-0.86,-1.02) rectangle (0.86,1.02);
			\foreach \v in {1,...,6} \node[bv,label={[label distance=-0.5mm]{90-60*(\v-1)}:\scriptsize$\v$}] (v\v) at ({90-60*(\v-1)}:0.55) {};
			\draw (v1) -- (v4) (v2) -- (v3) (v5) -- (v6);
		\end{scope}
		\begin{scope}[shift={(3.900,-2.20)}]
			\fill[black!7,rounded corners=4pt] (-0.86,-1.02) rectangle (0.86,1.02);
			\foreach \v in {1,...,6} \node[bv,label={[label distance=-0.5mm]{90-60*(\v-1)}:\scriptsize$\v$}] (v\v) at ({90-60*(\v-1)}:0.55) {};
			\draw (v1) -- (v4) (v2) -- (v5) (v3) -- (v6);
		\end{scope}
		\begin{scope}[shift={(5.850,-2.20)}]
			\fill[black!7,rounded corners=4pt] (-0.86,-1.02) rectangle (0.86,1.02);
			\foreach \v in {1,...,6} \node[bv,label={[label distance=-0.5mm]{90-60*(\v-1)}:\scriptsize$\v$}] (v\v) at ({90-60*(\v-1)}:0.55) {};
			\draw (v1) -- (v4) (v2) -- (v6) (v3) -- (v5);
		\end{scope}
		\begin{scope}[shift={(7.800,-2.20)}]
			\fill[black!7,rounded corners=4pt] (-0.86,-1.02) rectangle (0.86,1.02);
			\foreach \v in {1,...,6} \node[bv,label={[label distance=-0.5mm]{90-60*(\v-1)}:\scriptsize$\v$}] (v\v) at ({90-60*(\v-1)}:0.55) {};
			\draw (v1) -- (v5) (v2) -- (v3) (v4) -- (v6);
		\end{scope}
		\begin{scope}[shift={(0.000,-4.40)}]
			\fill[black!7,rounded corners=4pt] (-0.86,-1.02) rectangle (0.86,1.02);
			\foreach \v in {1,...,6} \node[bv,label={[label distance=-0.5mm]{90-60*(\v-1)}:\scriptsize$\v$}] (v\v) at ({90-60*(\v-1)}:0.55) {};
			\draw (v1) -- (v5) (v2) -- (v4) (v3) -- (v6);
		\end{scope}
		\begin{scope}[shift={(1.950,-4.40)}]
			\fill[black!7,rounded corners=4pt] (-0.86,-1.02) rectangle (0.86,1.02);
			\foreach \v in {1,...,6} \node[bv,label={[label distance=-0.5mm]{90-60*(\v-1)}:\scriptsize$\v$}] (v\v) at ({90-60*(\v-1)}:0.55) {};
			\draw (v1) -- (v5) (v2) -- (v6) (v3) -- (v4);
		\end{scope}
		\begin{scope}[shift={(3.900,-4.40)}]
			\fill[black!7,rounded corners=4pt] (-0.86,-1.02) rectangle (0.86,1.02);
			\foreach \v in {1,...,6} \node[bv,label={[label distance=-0.5mm]{90-60*(\v-1)}:\scriptsize$\v$}] (v\v) at ({90-60*(\v-1)}:0.55) {};
			\draw (v1) -- (v6) (v2) -- (v3) (v4) -- (v5);
		\end{scope}
		\begin{scope}[shift={(5.850,-4.40)}]
			\fill[black!7,rounded corners=4pt] (-0.86,-1.02) rectangle (0.86,1.02);
			\foreach \v in {1,...,6} \node[bv,label={[label distance=-0.5mm]{90-60*(\v-1)}:\scriptsize$\v$}] (v\v) at ({90-60*(\v-1)}:0.55) {};
			\draw (v1) -- (v6) (v2) -- (v4) (v3) -- (v5);
		\end{scope}
		\begin{scope}[shift={(7.800,-4.40)}]
			\fill[black!7,rounded corners=4pt] (-0.86,-1.02) rectangle (0.86,1.02);
			\foreach \v in {1,...,6} \node[bv,label={[label distance=-0.5mm]{90-60*(\v-1)}:\scriptsize$\v$}] (v\v) at ({90-60*(\v-1)}:0.55) {};
			\draw (v1) -- (v6) (v2) -- (v5) (v3) -- (v4);
		\end{scope}
		\foreach \r [count=\i from 0] in {{1,2,3,4,5,6},{1,2,3,4,5,6},{1,2,3,4,5,6},{2,1,4,3,6,5},{2,1,5,6,3,4},{2,1,6,5,4,3},{3,4,1,2,6,5},{3,5,1,6,2,4},{3,6,1,5,4,2},{4,3,2,1,6,5},{4,5,6,1,2,3},{4,6,5,1,3,2},{5,3,2,6,1,4},{5,4,6,2,1,3},{5,6,4,3,1,2},{6,3,2,5,4,1},{6,4,5,2,3,1},{6,5,4,3,2,1}}
			\foreach \v [count=\j from 0] in \r
				{\ifnum\i<3 \fill[identityblockshade] (9.55+0.3578*\j,1.02-0.3578*\i) rectangle ++(0.3578,-0.3578);\fi
				\draw[thin] (9.55+0.3578*\j,1.02-0.3578*\i) rectangle ++(0.3578,-0.3578) node[midway] {\footnotesize\v};}
		\draw[very thick] (9.55,-0.0533) -- ++(2.1467,0);
	\end{tikzpicture}
	\caption{The $15=\sfactexcludepars{5}$ perfect matchings of~$K_6$, forming a hyperfactorization of index~$1$.
		By \cref{thm:hyperfactorization-bound}, read as the involutory derangements of~$\symmetricgroup{6}$ and stacked with $r=1\cdot(6-3)=3$ copies of~$\identity$, they yield the $\URS{6}{3}$ shown on the right, whose first three rows are the copies of~$\identity$, followed by one row per matching, read row-wise from the top left.\label{fig:k6-perfect-matchings}}
\end{figure}
However, the upper bound on $\kURS{n}$ resulting from \cref{cor:involutory-derangement-construction} is superexponential in $n$.
Since $n-3$ is odd, the multiplicity $r=c(n-3)$ in \cref{thm:hyperfactorization-bound} is odd precisely when $c$ is odd; hence hyperfactorizations of odd index yield upper bounds on $\kURS{n}$.
We are therefore interested in hyperfactorizations of smaller odd index, yielding smaller multiplicities.
\begin{corollary}\label{cor:hyperfactorization-bound-all-n}
	Let $n\geq10$ be even and write $n=2^a+2d$ with $a:=\lfloor\log_2(n-2)\rfloor$, so that $1\leq d\leq 2^{a-1}$.
	Then
	\begin{equation*}
		\kURS{n}\leq\prod_{i=0}^{d-1}(n-3-2i)\leq(n-3)^d.
	\end{equation*}
	In particular, $\kURS{2^a+2}\leq 2^a-1$ for all $a\geq3$, and along every sequence $n=2^a+2d$ with fixed offset~$d$ the bound is polynomial of degree~$d$ in~$n$.
\end{corollary}
\begin{proof}
	As $a\geq3$, the hyperovals of the projective planes of order~$2^a$ yield hyperfactorizations of index~$1$ of $K_{2^a+2}$~\cite{cameron1976parallelisms}; see also~\cite[p.~797]{bamberg2026association}.
	A hyperfactorization of index~$c$ of $K_m$ yields one of index~$c(m-3)$ of $K_{m+2}$~\cite[Theorem~2.1]{boros1991existence}.
	Applying this $d-1$ times, each time multiplying the index by an odd number, gives a hyperfactorization of $K_n$ of odd index $c=\prod_{i=1}^{d-1}(n-3-2i)$.
	By \cref{thm:hyperfactorization-bound} there exists an $\URS{n}{c(n-3)}$, and $c(n-3)$ is odd and the product of $d$ factors not exceeding $n-3$.
\end{proof}
For instance, \cref{cor:hyperfactorization-bound-all-n} improves the bound of \cref{cor:involutory-derangement-construction} from $105$ to $7$ for $n=10$.
For $n$ far from the greatest power of two not exceeding~$n$, e.g., $n=2^a+(2^a-2)$ with $d=2^{a-1}-1=(n-2)/4$, the bound remains, however, of order $n^{\Theta(n)}$, as does~$\sfact{n-3}$.
Let us examine in the subsequent \cref{pro:urs-n-one-iff-power-of-two} the existence for $\lambda=1$, i.e., the smallest possible odd multiplicity.
In the following we show that in this setting the property of being pairwise reflection-symmetric is equivalent to being a \emph{$2$-subsquare complete} Latin square.
\begin{definition}[{\cite{heinrich1981maximum,hiner1970subsquare}, \cite[Section~1.6]{keedwell2015latin}}]\label{def:two-subsquare-complete}
	Let $A\in\range{n}^{n\times n}$ be a Latin square.
	An \emph{intercalate} of $A$ is a set of four index-pairs $(i_1,j_1),(i_1,j_2),(i_2,j_1),(i_2,j_2)$ with $i_1\neq i_2$ and $j_1\neq j_2$ inducing
	on $A$ a $2\times2$ Latin subsquare, i.e., $A_{i_1,j_1}=A_{i_2,j_2}$ and $A_{i_1,j_2}=A_{i_2,j_1}$; by \emph{induce} we mean, technically, up to bijective relabeling to become a matrix in $\range{2}^{2\times 2}$.
	We say that $A$ is \emph{$2$-subsquare complete} if any two distinct cells carrying the same symbol occur as two of the four corners of a common intercalate.
\end{definition}
\begin{lemma}\label{lem:two-subsquare-complete-iff-pairwise-reflection-symmetric}
	We have $A\in\URS{n}{1}$ iff $A$ is a $2$-subsquare complete Latin square.
\end{lemma}
\begin{proof}
	Denote by $A_{\cdot,\{j_1,j_2\}}$ the $n\times2$ submatrix of $A$ spanned by the columns~$j_1\neq j_2$.
	Pairwise reflection-symmetry of $A$ applies precisely when we have for all $j_1\neq j_2$ and all $p\neq q$ that
	\begin{equation}
		(p,q)\text{ is a row of }A_{\cdot,\{j_1,j_2\}}\text{ iff }(q,p)\text{ is a row of }A_{\cdot,\{j_1,j_2\}}.\label{eq:pairwise-reflection-symbol-in-unit-latin-squares}
	\end{equation}
	Let us show first that~\eqref{eq:pairwise-reflection-symbol-in-unit-latin-squares} implies $2$-subsquare completeness.
	Take two distinct cells $(i_1,j_1)$ and $(i_2,j_2)$ carrying the same symbol~$p=A_{i_1,j_1}=A_{i_2,j_2}$.
	We argue that the four index pairs $(i_1,j_1)$, $(i_1,j_2)$, $(i_2,j_1)$, and $(i_2,j_2)$ span an intercalate; for this it remains to show that $A_{i_2,j_1}=A_{i_1,j_2}=q$, for some $q\in\range{n}$.
	If $A_{i_2,j_1}=q$ (meaning that the $i_2$-th row of $A_{\cdot,\{j_1,j_2\}}$ reads $(q,p)$) then by pairwise reflection-symmetry also the row~$(p,q)$ must be present in $A_{\cdot,\{j_1,j_2\}}$.
	However, there is just one row starting with $p$ in $A_{\cdot,\{j_1,j_2\}}$ and this row is the one with index~$i_1$.
	Consequently, $A_{i_2,j_1}=A_{i_1,j_2}=q$.

	For the converse proof direction consider any row of $A_{\cdot,\{j_1,j_2\}}$, say indexed $i_1$ and reading $(p,q)$.
	Let $i_2$ be the row-index in which $p$ is hosted in column~$j_2$.
	Then $(i_1,j_1)$ and $(i_2,j_2)$ carry the same symbol~$p$, and there must be an intercalate (necessarily $\{(i_1,j_1),(i_1,j_2),(i_2,j_1),(i_2,j_2)\}$) inducing a $2\times 2$ Latin square; this means that the row~$i_2$ reads $(q,p)$.
\end{proof}
\begin{proposition}\label{pro:urs-n-one-iff-power-of-two}
	An $\URS{n}{1}$ exists iff $n=2^k$ for some $k\in\mathbb{N}\cup\{0\}$.
	This means $\kURS{2^k}=1$ for every $k\in\mathbb{N}$; recall that $\kURS{n}$ is defined for even $n$ only.
\end{proposition}
\begin{proof}
	As found by~\cite{heinrich1981maximum,hiner1970subsquare} (see also~\cite[p.~34]{keedwell2015latin}), every $2$-subsquare complete Latin square is, up to isotopy, the Cayley table of $(\mathbb{Z}_2^k,+)$ for some $k\in\mathbb{N}\cup\{0\}$.
	By \cref{lem:two-subsquare-complete-iff-pairwise-reflection-symmetric} this must consequently be the case for $\URS{n}{1}$, and therefore, $n=2^k$.
	Conversely, for $k\geq2$ the Cayley table of the non-cyclic group $(\mathbb{Z}_2^k,+)$, all of whose nonidentity elements have order~$2$, is $2$-subsquare complete by~\cite[Theorem~1.6.10]{keedwell2015latin}, and thus a member of $\URS{2^k}{1}$ by \cref{lem:two-subsquare-complete-iff-pairwise-reflection-symmetric}.
	Instead, for $k\leq1$ the claim is immediate.
\end{proof}
In particular, the $\URS{6}{3}$ of \cref{fig:k6-perfect-matchings} is optimal, i.e., $\kURS{6}=3$, as no $\URS{6}{1}$ exists by \cref{pro:urs-n-one-iff-power-of-two}; this shows that duplicate rows do not render impossible the existence of an optimal $\URS{n}{\kURS{n}}$.
\subsection{A direct product construction}\label{sec:direct-product-construction}
Given two matrices,
$A \in \URS{n}{\lambda}$ of size~$\lambda n \times n$ and $\tilde{A}\in\URS{\tilde{n}}{\tilde{\lambda}}$ of size~$\tilde{\lambda}\tilde{n}\times\tilde{n}$, we consider the \emph{direct product} $A \boxtimes\tilde{A}\in (\range{n}\times\range{\tilde{n}})^{\lambda n\cdot\tilde{\lambda}\tilde{n}\times n\tilde{n}}$ with entries specified by
\begin{equation*}
	(A\boxtimes\tilde{A})_{(i,\tilde{i}),(j, \tilde{j})}:=\left(A_{i, j},\tilde{A}_{\tilde{i},\tilde{j}}\right),
\end{equation*}
where rows are indexed by pairs $(i,\tilde{i})\in\range{\lambda n}\times \range{\tilde{\lambda}\tilde{n}}$ and columns by pairs $(j, \tilde{j}) \in \range{n}\times\range{\tilde{n}}$, both assumed, e.g., in lexicographical order.
Each row of $A\boxtimes\tilde{A}$ corresponds to the permutation~$(\sigma,\tilde{\sigma}):(j,\tilde{j})\mapsto (\sigma(j),\tilde{\sigma}(\tilde{j}))$, where $\sigma$, $\tilde{\sigma}$ are the rows of $A$, $\tilde{A}$ at indices $i$, $\tilde{i}$ read as permutations.
\begin{theorem}\label{thm:urs-product}
	Let $A\in\URS{n}{\lambda}$ and $\tilde{A}\in\URS{\tilde{n}}{\tilde{\lambda}}$.
	Then $A\boxtimes\tilde{A}\in\URS{n\tilde{n}}{\lambda\tilde{\lambda}}$, up to identification of the entries with $\range{n\tilde{n}}$-valued scalars.
\end{theorem}
\begin{proof}
	By reading, as above, each row of $A \boxtimes \tilde{A}$ as the map~$(\sigma, \tilde{\sigma})$, we notice that $(j,\tilde{j})\mapsto (\sigma(j),\tilde{\sigma}(\tilde{j}))$ is indeed a permutation of $\range{n}\times\range{\tilde{n}}$:
	An arbitrary $(g,\tilde{g})\in\range{n}\times\range{\tilde{n}}$ is namely the image of $(\sigma^{-1}(g), \tilde{\sigma}^{-1}(\tilde{g}))$ subject to the mapping $(\sigma,\tilde{\sigma})$.

	To show pairwise reflection-symmetry, fix two columns~$(j_1,\tilde{j_1}),(j_2, \tilde{j_2})$ of $A\boxtimes\tilde{A}$ and two symbols~$(p, \tilde{p})$, $(q, \tilde{q})$.
	Write
	\begin{align*}
		\bifMatfree{(j_1,\tilde{j_1})}{(j_2,\tilde{j_2})}{(p,\tilde{p})}{(q,\tilde{q})} & =\card{\{(i,\tilde{i}):A_{i,j_1}=p,A_{i,j_2}=q\}\cap \{(i,\tilde{i}):\tilde{A}_{\tilde{i},\tilde{j_1}}=\tilde{p},\tilde{A}_{\tilde{i},\tilde{j_2}}=\tilde{q}\}}                                                                                                          \\
		                                                                                & =\bifMatfree{j_1}{j_2}{p}{q}\cdot\bifMatfree{\tilde{j_1}}{\tilde{j_2}}{\tilde{p}}{\tilde{q}}=\bifMatfree{j_1}{j_2}{q}{p}\cdot\bifMatfree{\tilde{j_1}}{\tilde{j_2}}{\tilde{q}}{\tilde{p}}=\bifMatfree{(j_1,\tilde{j_1})}{(j_2,\tilde{j_2})}{(q,\tilde{q})}{(p,\tilde{p})}
	\end{align*}
	showing symmetry of the bivariate frequency matrix.

	After bijectively relabeling all matrix entries with symbols from $\range{n\tilde{n}}$ and recalling that the matrix has a row-count of $\lambda\tilde{\lambda}n\tilde{n}$, we can apply \cref{lem:family-of-permutations-plus-reflection-symmetric-implies-urs} and obtain that each symbol~$q\in\range{n\tilde{n}}$ appears precisely $\lambda\tilde{\lambda}$ times in each column.
\end{proof}
Applying the direct product theorem with \cref{pro:urs-n-one-iff-power-of-two} transfers any upper bound on $\kURS{m}$ to all $n=2^k m$.
\begin{corollary}\label{cor:product-bound-infinite-family-constant-kurs}
	Let $m$ be even and $k\geq0$.
	Then
	\begin{equation*}
		\kURS{2^k m}\leq\kURS{m}.
	\end{equation*}
	In particular, the sequence $(\kURS{2^k m})_{k\geq0}$ is uniformly upper bounded by a constant independent of $k$.
\end{corollary}
\begin{proof}
	By \cref{pro:urs-n-one-iff-power-of-two} there is an $A\in\URS{2^k}{1}$, and by definition of $\kURS{m}$ there is an $\tilde{A}\in\URS{m}{\kURS{m}}$.
	By \cref{thm:urs-product}, $A\boxtimes\tilde{A}\in\URS{2^k m}{\kURS{m}}$; as $\kURS{m}$ is odd, $\kURS{2^k m}\leq\kURS{m}$.
\end{proof}

\subsection{Counting pairwise reflection-symmetric matrices}\label{sec:counting-urs}
In the following we give some insights on the quantity of pairwise reflection-symmetric exemplars that certain $(n,\lambda)$-constellations admit.
For this purpose, we report results relying on the enumeration carried out by traversing the complete feasible region of a Quadratically Constrained Program (QCP) modeling the pairwise reflection-symmetric property; it is a binary QCP relying on a one-hot encoding of the symbol in each matrix entry which can be found in the appendix, \cref{sec:utilized-qcp}.
The QCP is generated in \texttt{Julia} (version 1.12.5) and solved using the Mixed Integer Programming solver Gurobi~\cite{gurobi2025manual} (version 12.0.3) by setting \texttt{PoolSearchMode=2} and \texttt{PoolSolutions=200000} to carry out the enumeration of all feasible points; four threads are made available to Gurobi with a wallclock time limit of 16 hours, i.e., approximately $5.8\cdot10^4$ seconds.
As the QCP has no objective, every feasible point is optimal. If Gurobi terminates with status \texttt{OPTIMAL} (resp.\ \texttt{INFEASIBLE}) under \texttt{PoolSearchMode=2}, the pool is guaranteed to contain the \texttt{PoolSolutions} best solutions. Hence, whenever fewer than $200\,000$ solutions are returned before reaching the time limit, the enumeration is exhaustive. This is the case for all instances not marked in italic; for the italic instances the time limit was reached, and the reported counts are lower bounds.
The experiments have been run on a hardware equipped with an AMD Ryzen 9 5900X 12-Core Processor (3.80 GHz) and a maximum of 125.7 GiB of RAM; the results are reported in \cref{tab:computational-counting-results}.
On those instances for which the enumeration terminated within the time limit (i.e., those not marked in italic), the reported exact counts, including the zero counts of the infeasible instances, are independently reproduced by the exhaustive solution counting of the open-source Constraint Integer Programming framework SCIP~\cite{hojny2025scip} (version 10.0.3), whose linear programming relaxations are in turn solved by its component SoPlex (version 8.0.3).
The source code for reproducing the computational results is available online; see Section~{\hypersetup{linkcolor=.}\hyperref[sec:data-availability]{\emph{Data availability}}}.
\begin{table}[tb]
	\centering
	\caption{Reduced $\URS{n}{\lambda}$ enumeration.
		\emph{Time to first}:
		seconds until the first feasible representative is found, ``/'' if none exists (or none was encountered within the time limit).
		\emph{Enumeration time}:
		seconds for the full enumeration, or for the infeasibility certificate when no representative exists.
		\emph{Representatives}:
		number of reduced representatives.
		Entries in italic mark a hit of the time limit.\label{tab:computational-counting-results}}
	\setlength{\tabcolsep}{4pt}
	
	\bigskip
	
	\begin{tabular}{l|rrrrrrrr}
		$(n,\lambda)$ & $(4,1)$ & $(6,1)$                      & $(8,1)$                             & $(10,1)$                            & $(12,1)$ & $(3,2)$ & $(4,2)$                             & $(5,2)$                             \\[3pt]\hline
		Time to first [s]    & 0.001 & /                                                 & 0.01                                                     & /                                                        & /     & 0.001 & 0.002                                                    & 0.005                                                    \\
		Enumeration time [s] & 0.001 & 0.002                                             & 0.04                                                     & 4.52                                                     & 0.27  & 0.001 & 0.004                                                    & 0.08                                                     \\
		Representatives      & 1     & 0                                                 & 30                                                       & 0                                                        & 0     & 1     & 6                                                        & 6                                                        \\[2.5ex]
		$(n,\lambda)$ & $(6,2)$ & $(7,2)$                      & $(8,2)$                             & $(9,2)$                             & $(4,3)$  & $(6,3)$ & $(8,3)$                             & $(10,3)$                            \\[3pt]\hline
		Time to first [s]    & 0.01  & 0.04                                              & 0.10                                                     & 47.01                                                    & 0.003 & 49.05 & 505                                                      & /                                                        \\
		Enumeration time [s] & 86.04 & 1.7\,\textperiodcentered{}\,10\textsuperscript{4} & \emph{5.8\,\textperiodcentered{}\,10\textsuperscript{4}} & \emph{5.8\,\textperiodcentered{}\,10\textsuperscript{4}} & 0.03  & 3597  & \emph{5.8\,\textperiodcentered{}\,10\textsuperscript{4}} & \emph{5.8\,\textperiodcentered{}\,10\textsuperscript{4}} \\
		Representatives      & 600   & 120                                               & \ensuremath{\geq}\,20\,710                                 & \ensuremath{\geq}\,24                                    & 21    & 1051  & \ensuremath{\geq}\,8783                                  & \ensuremath{\geq}\,0                                     \\
	\end{tabular}
\end{table}
The results are in line with the finding in \cref{pro:urs-n-one-iff-power-of-two} for $\lambda=1$ and $n\in\{4,6,\dots,12\}$, which can be solved in milliseconds up to a couple of seconds.
Interestingly, infeasibility of the QCP for $n=12$ is shown approximately 17 times faster than the one for the considerably smaller QCP for $n=10$.
Feasibility remains unanswered only for $(n,\lambda)=(10,3)$ while the exact counts could not be determined within the time limit for $\lambda\geq2$ and $n\geq8$.
\begin{remark}
	In passing, following up on \cref{obs:urs-is-strength-two},  we straightforwardly extended the enumeration algorithm to count pairwise reflection-symmetric matrices being a PSCA of strength $k$.
	With this adaptation we found that there is also a strength-$5$ PSCA in $\range{6}^{120\times 6}$ of this particular kind contributing the new pair $(n,k)=(6,5)$ to the list~\cite{iurlano2023growth} of such $(n,k)$-constellations; see the appendix.
\end{remark}
\begin{remark}\label{rem:llm-usage-attribution}
	The pattern for $\lambda=2$ and prime numbers $n$ seems to support the hypothesis of a count of $(n-2)!$ solutions, which motivated an investigation by interacting with the Large Language Model (LLM) Opus 4.7 via the agent Claude~\cite{anthropic2026claude}.
	We provided a file with the enumerated solutions to the LLM and asked for pattern recognition.
	Also the code for the QCP was enclosed.
	After several interactions and requests for generalizing the insights, some claimed proofs were proposed by the LLM, which we could verify (modulo slight corrections needed) and which are reported in edited form in the remainder of this section.
	Notably, non-negligible time was spent to replace some intermediate results of the LLM's solution by suitable already known ones.
\end{remark}
We now give a group-theoretic construction of members of $\URS{n}{\lambda}$ and derive from it an explicit lower bound on $\card{\URSreduced{n}{2}}$, the number of reduced elements in $\URS{n}{2}$.
\begin{definition}[Matrix induced by a permutation group]\label{def:group-induced-matrix}
	For a subgroup~$G\leq\symmetricgroup{n}$ with $\card{G}=\lambda n$, let $\matrixrowsGroupInduced{G}\in\range{n}^{\lambda n\times n}$ denote the matrix whose rows are indexed by the elements~$g\in G$ and whose columns are indexed by $j\in\range{n}$, with entries~$\bigl(\matrixrowsGroupInduced{G}\bigr)_{g,j}:=g(j)$, the rows being arranged in lexicographical order of the tuples $(g(1),\ldots,g(n))$.
\end{definition}
Since each $g$ is a bijection of $\range{n}$, every row of $\matrixrowsGroupInduced{G}$ is a permutation of $\range{n}$.
As $\identity\in G$ yields the lexicographically smallest row $(1,2,\ldots,n)$, the matrix $\matrixrowsGroupInduced{G}$ is reduced (\cref{def:reduced-frequency-rectangle}) whenever it is a frequency rectangle.
\begin{definition}[(Generous) transitivity {\cite[p.~26]{godsil2001algebraic}}]\label{def:transitivity-generous-transitivity}
	The action of $G\leq\symmetricgroup{n}$ is \emph{transitive} if for every pair of points~$x,y\in\range{n}$ there exists $g\in G$ with $g(x)=y$.
	It is \emph{generously transitive} if for every pair of \emph{distinct} points~$x,y\in\range{n}$ there exists $g\in G$ with $g(x)=y$ \emph{and} $g(y)=x$, i.e., some element of $G$ interchanges $x$ and $y$.
\end{definition}
Generous transitivity implies transitivity, since an element interchanging $x$ and $y$ in particular maps $x$ to $y$.
Note that for $\matrixrowsGroupInduced{G}\in\URS{n}{\lambda}$ the group~$G$ must clearly be transitive, already due to the presence of each symbol within each column.
This can be sharpened as follows.
\begin{theorem}\label{thm:generously-transitive-construction}
	Let $G\leq\symmetricgroup{n}$ be a subgroup of order~$\card{G}=\lambda n$.
	Then $\matrixrowsGroupInduced{G}\in\URS{n}{\lambda}$ iff $G$ is generously transitive.
\end{theorem}
\begin{proof}
	As every row of $\matrixrowsGroupInduced{G}$ is a permutation of $\range{n}$, by \cref{lem:family-of-permutations-plus-reflection-symmetric-implies-urs}, membership in $\URS{n}{\lambda}$ is equivalent to pairwise reflection-symmetry.

	Fix positions $j\neq k$.
	If $G$ is generously transitive, there exists $h\in G$ interchanging $j$ and $k$.
	Since $(g\circ h)(j)=g(k)$ and $(g\circ h)(k)=g(j)$, the bijection $g\mapsto g\circ h$ of $G$ maps the elements realizing $(p,q)$ in positions $(j,k)$ onto those realizing $(q,p)$, so $\bifMatfree{j}{k}{p}{q}=\bifMatfree{j}{k}{q}{p}$ for all $p,q$, i.e., pairwise reflection-symmetry holds.
	Conversely, as $\identity\in G$ realizes $(j,k)$ in positions $(j,k)$, we have $\bifMatfree{j}{k}{j}{k}\geq1$; pairwise reflection-symmetry then gives $\bifMatfree{j}{k}{k}{j}\geq1$, i.e., some $h\in G$ satisfies $h(j)=k$ and $h(k)=j$.
	As $j\neq k$ were arbitrary, $G$ is generously transitive.
\end{proof}
\begin{definition}[Representation of the generalized dihedral group]\label{def:generalized-dihedral}
	Let $A$ be an abelian group and, for $c\in A$, let $\tau_c,\sigma_c\in S_A$ be
	\begin{equation*}
		\tau_c\colon a\mapsto a+c,\qquad\qquad \sigma_c\colon a\mapsto c-a.
	\end{equation*}
	Provided $A$ is not an elementary abelian $2$-group, let us regard $\Dih{A}$ as the subgroup of $S_A$ formed by the permutations $\{\tau_c:c\in A\}\cup\{\sigma_c:c\in A\}$ under function composition.
\end{definition}
The group $\Dih{A}$ is isomorphic to the \emph{generalized dihedral group} associated to $A$, which is classically defined as a semidirect product~\cite[p.~178]{dummit2004abstract}; when $A=\mathbb{Z}_n$ with $n\geq3$ it recovers the ordinary \emph{dihedral group} $D_n\leq\symmetricgroup{n}$ of symmetries of the regular polygon on $n$ vertices.

\begin{corollary}\label{cor:generalized-dihedral-is-urs}
	For every $n\geq3$ and every abelian group~$A$ of order~$n$ not being an elementary abelian $2$-group, we have~$\matrixrowsGroupInduced{\Dih{A}}\in\URS{n}{2}$.
\end{corollary}
\begin{proof}
	For any two points~$x,y\in A$ the reflection $a\mapsto x+y-a$ swaps $x$ and $y$, so $\Dih{A}$ is generously transitive.
	Since $A$ is not an elementary abelian $2$-group, inversion $a\mapsto-a$ is a nontrivial automorphism, so these reflections are involutions outside the rotation subgroup and $\Dih{A}$ has order~$2n$ in $\symmetricgroup{n}$.
	Thus \cref{thm:generously-transitive-construction} applies.
	Note that if $A$ would be an elementary abelian $2$-group, inversion would be trivial, the reflections would coincide with the rotations, and the natural action would collapse to the order-$n$ rotation subgroup, yielding only a member of $\URS{n}{1}$.
\end{proof}
Counting the embedded generalized dihedral groups turns \cref{cor:generalized-dihedral-is-urs} into an explicit lower bound on $\card{\URSreduced{n}{2}}$.
For $n\geq3$, let $L(n)$ denote the number of \emph{labeled abelian groups of order~$n$ with a fixed identity}, i.e., the abelian group structures on $\range{n}$ having a prescribed element as identity, equivalently the abelian Cayley tables of order~$n$.
Writing $\Aut{A}$ for the automorphism group of $A$, it equals $\sum_{\card{A}=n}(n-1)!/\card{\Aut{A}}$, the sum extending over all isomorphism types of abelian groups of order~$n$, and is enlisted as the sequence \href{https://oeis.org/A058162}{A058162}~\cite{oeis2026online}.
Let $L_2(n)$ be the number of labeled elementary abelian $2$-group structures on $\range{n}$ with the prescribed identity element, so $L_2(n)=0$ unless $n=2^k$, in which case~$L_2(n)=(n-1)!/\card{\GL{k}{2}}$ with $\card{\GL{k}{2}}=2^{\binom{k}{2}}\prod_{\ell=1}^k(2^\ell-1)$; e.g., see~\cite[p.~3]{rosenberg1973number}.
\begin{proposition}\label{pro:dihedral-lower-bound-labeled}
	For $n\geq3$, with $L(n)$ and $L_2(n)$ as above and $\URSreduced{n}{2}$ denoting the reduced elements in $\URS{n}{2}$, we have
	\begin{equation}
		\card{\URSreduced{n}{2}}\geq\card{\URSreduced{n}{2}\cap \{\matrixrowsGroupInduced{G}\}_{G\leq \symmetricgroup{n}}}\geq\begin{cases}L(n)-L_2(n) & \text{ if }n\in\{2^k\}_{k\in\mathbb{N},k\geq 2}, \\
             L(n)        & \text{ otherwise.}\end{cases}\label{eq:lower-bound-labeled}
	\end{equation}
\end{proposition}
\begin{proof}
	We argue directly with the $L(n)$ abelian group structures on $\range{n}$ that share a prescribed identity element $e\in\range{n}$, of which $L_2(n)$ are an elementary abelian $2$-group.
	For such a structure $(\range{n},+_A)$, whose carrier set we denote by $\range{n}_{+_A}$ to express the dependency on $+_A$, not an elementary abelian $2$-group, the dihedral realization in \cref{def:generalized-dihedral} gives an order-$(2n)$ subgroup of $\symmetricgroup{n}$ with $\matrixrowsGroupInduced{\Dih{A}}\in\URS{n}{2}$ by \cref{cor:generalized-dihedral-is-urs}.

	We show that the operation $+_A$ is recoverable from the permutation group $\Dih{\range{n}_{+_A}}$ alone, so that $\range{n}_{+_A}\mapsto\Dih{\range{n}_{+_A}}$ is injective.
	We argue that the rotations form a subgroup that coincides with $T:=\langle g\in\Dih{\range{n}_{+_A}}:\ord{g}>2\rangle$:
	Every element of order greater than $2$ is a rotation, giving $T\subseteq\{\tau_c\}_{c\in\range{n}_{+_A}}$; conversely, as $\range{n}_{+_A}$ is not an elementary abelian $2$-group, the rotations of order greater than $2$ \emph{generate} all of $\{\tau_c\}_{c\in \range{n}_{+_A}}$, since a rotation $t$ with $\ord{t}\leq2$ equals $x^{-1}(xt)$ for a fixed rotation $x$ with $\ord{x}>2$.
	The latter product consists of factors both of order greater than $2$; for this note that $(xt)^2=x^2\neq\identity$.
	The operation $+_A$ is then uniquely determined by $T$ and the neutral element~$e\in\range{n}_{+_A}$:
	For $a,b\in\range{n}$ one has $a+_A b=\tau(b)$, where $\tau\in T$ is the unique rotation with $\tau(e)=a$; the rotations are exactly the regular representation of $\range{n}_{+_A}$, which is by definition simply transitive.
	Hence $\Dih{\range{n}_{+_A}}=\Dih{\range{n}_{+_{A'}}}$ forces $+_A=+_{A'}$.

	Distinct structures therefore give distinct subgroups of $\symmetricgroup{n}$, hence distinct reduced representatives in $\URS{n}{2}$.
	Finally, an elementary abelian $2$-group has trivial inversion, so its reflections coincide with its rotations and $\matrixrowsGroupInduced{\Dih{A}}\in\URS{n}{1}$ instead; such structures occur precisely for $n=2^k$ and number $L_2(n)$.
	Thus $(\range{n},+_A)\mapsto\matrixrowsGroupInduced{\Dih{A}}$ yields $L(n)-L_2(n)$ distinct reduced members of $\URS{n}{2}$ when $n=2^k$ with $k\geq2$, and $L(n)$ otherwise, showing~\eqref{eq:lower-bound-labeled}.
\end{proof}

\section{Conclusion}\label{sec:conclusion}
We investigated which per-column multiplicities $\lambda$ admit pairwise reflection-symmetric frequency rectangles of order~$n$, finding a benign even regime and a subtle odd one.
For even $\lambda$ existence is always guaranteed (\cref{cor:even-lambda-setting-is-simple}), and already the case~$\lambda=2$ is structurally rich:
Generalized dihedral groups yield explicit constructions and a lower bound~\eqref{eq:lower-bound-labeled} on the number of solutions; we conjecture that it is tight for odd $n$, as observed for $n\in\{3,5,7\}$; already $n=9$, the first case admitting a non-cyclic abelian group, remains open.
The delicate quantity is the smallest \emph{odd} multiplicity~$\kURS{n}$ (for even $n$), whose behavior seems to be governed by how far the prime factorization of $n$ deviates from a pure power of two:
Powers of two $n=2^k$ attain the optimum $\kURS{n}=1$ through the Cayley table of $(\mathbb{Z}_2^k,+)$; a direct product construction propagates relatively economically this construction's lightness to selected $n$.
As a by-product, we demonstrated that pairwise reflection-symmetry forms a natural generalization of $2$-subsquare completeness.

While \cref{cor:hyperfactorization-bound-all-n} improves on the bound $\sfact{n-3}$ of \cref{cor:involutory-derangement-construction}, in the worst case it still grows as $n^{\Theta(n)}$.
Several follow-up questions arise:
Does a merely exponential, polynomial, or perhaps even constant bound on $\kURS{n}$ apply uniformly in~$n$?
Does relaxing feasibility to an optimization variant, minimizing the summed column-pair deviations from pairwise reflection-symmetry, turn $\kURS{n}$ into an insightful measure of obstruction when no truly feasible object exists?
Do the existence guarantees of \cref{cor:even-lambda-setting-is-simple,cor:involutory-derangement-construction} survive the additional demand of pairwise distinct rows (ignoring trivially large values of $\lambda$, for which this is impossible)?
A further direction would be to approach the count of pairwise reflection-symmetric matrices by estimating their information-theoretic density and an examination of whether their manifold symmetries can be captured by a natural compression scheme.

\section*{Acknowledgements}
We thank an anonymous referee for many helpful suggestions and for pointing out the connection to hyperfactorizations, which led to \cref{thm:hyperfactorization-bound} and generalizes our former construction (now \cref{cor:involutory-derangement-construction}).

\section*{Declaration on generative AI}
Opus 4.7 has been guided towards mathematical discovery where indicated; see \cref{rem:llm-usage-attribution} for further specifications.
Moreover, it has been used to check notation-consistency, spell-checking, grammar, and to affirm the absence of any potential errors on a superficial level (typos, small logical flaws in the presentation of the arguments) arising during the write-up.
After using these tools, the authors reviewed and edited the content as needed and take full responsibility for the publication's content.

\section*{Conflict of interest}
The authors declare no conflicts of interest.

\phantomsection\section*{Data availability}\label{sec:data-availability}

The source code for reproducing the computational enumerations is available as ancillary material of the preprint~\href{https://doi.org/10.48550/arXiv.2606.28315}{[arXiv.2606.28315]}.

\clearpage

\setcounter{section}{0}
\renewcommand{\thesection}{\Alph{section}}
\renewcommand{\thesubsection}{\thesection.\arabic{subsection}}
\renewcommand{\theHsection}{appendix.\Alph{section}}

\section{Appendix}

\subsection{The utilized QCP}\label{sec:utilized-qcp}
In what follows, we briefly explain how one can \emph{computationally enumerate} all $\URS{n}{\lambda}$ with $n\geq2$ which are \emph{reduced} (\cref{def:reduced-frequency-rectangle}) as the feasible region of a binary QCP.
For every triple $(i,j,s)\in\range{\lambda n}\times\range{n}\times\range{n}$ we introduce the decision variables $y_{i,j,s}\in\{0,1\}$ with the intended meaning $y_{i,j,s}=1$ iff $A_{ij}=s$.
The binary $y$ is decoded back into the matrix~$A$ via $A_{ij}=\sum_{s\in\range{n}}sy_{i,j,s}$.

\medskip

These variables will be subject to the following three types of constraints.

\medskip

\emph{Frequency rectangle constraints.}
Each entry contains exactly one symbol, and every row is a permutation of $\range{n}$:
\begin{align}
	\sum_{s\in\range{n}} y_{i,j,s} & = 1 &  & \forall i\in\range{\lambda n},j\in\range{n},\label{eq:onehot-cell}     \\
	\sum_{j\in\range{n}} y_{i,j,s} & = 1 &  & \forall i\in\range{\lambda n},s\in\range{n}.\label{eq:onehot-row-perm}
\end{align}
Column-uniformity $\colfreqvec{j}=(\lambda)_{q\in\range{n}}$, for all $j\in\range{n}$, will be automatically enforced by the symmetry-breaking constraint~\eqref{eq:onehot-col-one} in conjunction with the constraints~\eqref{eq:onehot-refl-sym} responsible for pairwise reflection-symmetry, both stated below.

\medskip

\emph{Pairwise reflection-symmetry.}
For every column pair $(j_1,j_2)$ with $j_1<j_2$ and every symbol pair $(p,q)\in\range{n}^2$ with $p<q$,
\begin{align}
	\sum_{i\in\range{\lambda n}} y_{i,j_1,p}y_{i,j_2,q}=\sum_{i\in\range{\lambda n}} y_{i,j_1,q}y_{i,j_2,p}.\label{eq:onehot-refl-sym}
\end{align}
(These quadratic constraints can technically be linearized via McCormick envelopes.)

\medskip

\emph{Symmetry breaking.}
We partition the rows into $n$ blocks of $\lambda$ consecutive rows; block $b\in\range{n}$ consists of rows~$(b-1)\lambda+1,\ldots,b\lambda$.
We fix the first column to be the block index by
\begin{align}
	y_{i,1,\lceil i/\lambda\rceil}=1\qquad\forall i\in\range{\lambda n},\label{eq:onehot-col-one}
\end{align}
require the rows within a block to be lexicographically non-decreasing on the remaining columns,
\begin{align}
	\sum_{s\in\range{n}} sy_{i,2,s}                              & \leq\sum_{s\in\range{n}} sy_{i+1,2,s},\notag                                                          \\
	\sum_{j=2}^{n-1}\sum_{s\in\range{n}} s(n+1)^{n-1-j}y_{i,j,s} & \leq\sum_{j=2}^{n-1}\sum_{s\in\range{n}} s(n+1)^{n-1-j}y_{i+1,j,s},\label{eq:onehot-lex-within-block}
\end{align}
for every $i=(b-1)\lambda+k$ with $b\in\range{n}$ and $k\in\range{\lambda-1}$, and finally fix the first row of $A$ to be the identity permutation and the leading row of the second block to start with $(2,1)$:
\begin{align}
	y_{1,j,j}=1\quad\forall j\in\range{n}, \qquad y_{\lambda+1,2,1}=1.\label{eq:onehot-sym-break}
\end{align}

\clearpage

\subsection{A certifying instance}\label{sec:certifying-instance}
\lstset{caption={Certifying an $\URS{6}{20}$ that simultaneously is a PSCA of strength $5$.}}
\lstdefinelanguage{JuliaMinimalistic}{
	keywords={using, @show, @assert, all, for, in, length},
	sensitive=false,
	comment=[l]{\#},
	morestring=[b]"}
\begin{lstlisting}[basicstyle=\ttfamily\selectfont\scriptsize, keywordstyle=\ttfamily\selectfont\bfseries\color{blue}, commentstyle=\ttfamily\selectfont\color{red!60!orange}, tabsize=3,frame=none, backgroundcolor = \color{white}, numbers=none,inputencoding=utf8, mathescape=false, language=JuliaMinimalistic, xleftmargin=0mm,showstringspaces=false, breaklines=true, columns=fullflexible]
#!/usr/bin/env julia

A =	[
				[1,2,3,4,5,6], [2,1,4,3,5,6], [3,1,4,2,5,6], [4,1,2,5,3,6], [5,1,2,4,3,6], [6,1,2,4,3,5],
				[1,2,3,6,5,4], [2,1,5,3,4,6], [3,1,5,2,4,6], [4,1,3,5,2,6], [5,1,3,4,2,6], [6,1,3,5,2,4],
				[1,2,4,6,5,3], [2,1,6,3,5,4], [3,1,6,2,4,5], [4,1,6,2,3,5], [5,1,6,3,2,4], [6,1,4,2,5,3],
				[1,2,5,6,4,3], [2,1,6,4,5,3], [3,1,6,5,4,2], [4,1,6,5,3,2], [5,1,6,4,2,3], [6,1,5,3,4,2],
				[1,3,2,5,4,6], [2,3,1,5,6,4], [3,2,1,4,6,5], [4,2,1,5,6,3], [5,2,1,4,6,3], [6,2,3,1,4,5],
				[1,3,2,6,4,5], [2,3,4,6,5,1], [3,2,4,1,5,6], [4,2,3,1,6,5], [5,2,3,4,1,6], [6,2,3,5,4,1],
				[1,3,4,6,5,2], [2,3,5,1,4,6], [3,2,5,6,4,1], [4,2,3,5,6,1], [5,2,3,6,1,4], [6,2,4,1,5,3],
				[1,3,5,6,4,2], [2,3,6,4,1,5], [3,2,6,5,1,4], [4,2,6,5,1,3], [5,2,6,4,1,3], [6,2,5,1,4,3],
				[1,4,2,6,3,5], [2,4,1,3,6,5], [3,4,1,2,6,5], [4,3,1,5,6,2], [5,3,1,4,6,2], [6,3,2,1,5,4],
				[1,4,3,6,2,5], [2,4,5,3,1,6], [3,4,5,1,6,2], [4,3,2,5,1,6], [5,3,2,1,6,4], [6,3,2,4,5,1],
				[1,4,5,3,2,6], [2,4,5,6,1,3], [3,4,5,2,6,1], [4,3,2,6,1,5], [5,3,2,4,6,1], [6,3,4,1,5,2],
				[1,4,5,6,2,3], [2,4,6,3,1,5], [3,4,6,2,1,5], [4,3,6,5,1,2], [5,3,6,4,1,2], [6,3,5,1,4,2],
				[1,5,2,6,3,4], [2,5,1,3,6,4], [3,5,1,2,6,4], [4,5,1,2,6,3], [5,4,1,3,6,2], [6,4,2,1,3,5],
				[1,5,3,6,2,4], [2,5,4,1,6,3], [3,5,4,2,1,6], [4,5,2,1,3,6], [5,4,2,6,3,1], [6,4,3,1,2,5],
				[1,5,4,2,3,6], [2,5,4,3,6,1], [3,5,4,6,1,2], [4,5,3,6,2,1], [5,4,3,1,2,6], [6,4,5,1,3,2],
				[1,5,4,6,3,2], [2,5,6,3,1,4], [3,5,6,2,1,4], [4,5,6,3,1,2], [5,4,6,2,1,3], [6,4,5,2,3,1],
				[1,6,2,5,3,4], [2,6,1,3,4,5], [3,6,1,2,5,4], [4,6,1,3,2,5], [5,6,1,2,3,4], [6,5,2,1,3,4],
				[1,6,3,4,2,5], [2,6,1,5,4,3], [3,6,1,4,5,2], [4,6,1,5,2,3], [5,6,1,4,3,2], [6,5,3,1,2,4],
				[1,6,4,3,5,2], [2,6,4,3,5,1], [3,6,4,2,5,1], [4,6,2,5,3,1], [5,6,2,4,3,1], [6,5,4,1,2,3],
				[1,6,5,2,4,3], [2,6,5,3,4,1], [3,6,5,2,4,1], [4,6,3,5,2,1], [5,6,3,4,2,1], [6,5,4,3,2,1]
			]


using Combinatorics

n = length(A[1]); L = length(A)                  # A has L rows (indexed i) and n columns (indexed j)
pos(i, q) = findfirst(==(q), A[i])               # column in which symbol q occurs in row i

# strength-(n-1) PSCA: every (sigma_1,...,sigma_{n-1}) occurs as a subsequence of exactly L/(n-1)! rows
is_psca = all(count(i -> all(pos(i, sigma[m]) < pos(i, sigma[m+1]) for m in 1:n-2), 1:L)
						== div(L, factorial(n-1))          for sigma in permutations(1:n, n-1))

# pairwise reflection-symmetry: for all columns j1 < j2 and symbols p < q, (p,q) and (q,p) occur equally often
is_prs = all(count(i -> A[i][j1]==p && A[i][j2]==q, 1:L) == count(i -> A[i][j1]==q && A[i][j2]==p, 1:L)
						for (j1,j2) in combinations(1:n,2) for (p,q) in combinations(1:n,2))

# uniform columns: each symbol q occurs L/n times in each column j; this is a redundant double-check
is_uniform = all(count(i -> A[i][j]==q, 1:L) == div(L,n) for j in 1:n for q in 1:n)

@show(is_psca, is_prs, is_uniform);

# Output:
# is_psca = true
# is_prs = true
# is_uniform = true
\end{lstlisting}
\end{document}